\documentclass[runningheads]{llncs}
\newtheorem{sentence}{Sentence}
\begin{document}
\title{Fast generalized Bruhat decomposition}
\titlerunning{Fast generalized Bruhat decomposition} %  Matrix decomposition in parallel computer algebra}
\author{Gennadi Malaschonok
   \thanks{ Preprint of the paper: G. Malaschonok. Fast Generalized Bruhat Decomposition. Computer Algebra in Scientific Computing,
   LNCS 6244, Springer, Berlin, 2010. P. 194-202.
   (Supported by the Sci. Program Devel. Sci. Potent. High. School, RNP.2.1.1.1853.)}
}
\institute{Tambov State University, \\  Internatsionalnaya 33,
392622 Tambov, Russia}

\maketitle

\begin{abstract}
%\medskip
The deterministic recursive pivot-free algorithms for the computation of generalized Bruhat decomposition
of the matrix in the field and for the computation of the inverse matrix are
presented.
This method has the same complexity as algorithm of matrix multiplication and it is
suitable for the parallel computer systems.
\end{abstract}

%\keywords{generalized Bruhat decomposition, recursive algorithms, matrix inversion}

\section{Introduction}

An $LU$ matrix decomposition without pivoting is a decomposition of the form $A=LU$, a
decomposition with partial pivoting has the form $PA=LU$,
and decomposition with full pivoting (Trefethen and Bau) has the form $PAQ=LU$,
where $L$ and $U$ are lower and upper triangular matrices, $P$ and $Q$ is a permutation
matrix.

French mathematician  Francois Georges Ren\'e Bruhat was the first
who worked with matrix decomposition in the form  $A=VwU$, where
$V$ and $U$ are nonsingular upper triangular matrices and $w$ is a
matrix of permutation. Bruhat decomposition plays  an important
role in algebraic group. The generalized Bruhat decomposition was
introduced and developed by D.Grigoriev\cite{10},\cite{11}. He uses the
Bruhat decomposition in the form  $A=VwU$, where $V$ and $U$ are
upper triangular matrices but they may be singular when the matrix
$A$ is singular. In the papers \cite{12} and \cite{13} there was
analyzed the sparsity pattern of triangular factors of the Bruhat
decomposition of a nonsingular matrix over a field.

Fast matrix multiplication and fast block matrix inversion were
discovered by Strassen \cite{1}. The complexity of Strassen's
recursive algorithm for block matrix inversion is the same as the
complexity of an algorithm for matrix multiplication. But in this
algorithm it is assumed that principal minors are invertible and
leading elements are nonzero as in the most of  direct algorithms
for matrix inversion. There are known other recursive methods for
adjoint and inverse matrix computation, which  have the complexity
of matrix multiplications(\cite{4}-\cite{6}).

In a general case it is necessary to find  suitable nonzero
elements and to perform  permutations of  matrix columns or rows.
Bunch and Hopkroft suggested such algorithm with full pivoting for
matrix inversion \cite{7}.

The permutation operation is not a very difficult operation  in
the case of sequential computations by one processor, but it is a
difficult operation in the case of parallel computations, when
different blocks of a matrix are disposed in different processors.
A matrix decomposition without permutations is needed for parallel
computation for construction of efficient and fast computational
schemes.

The problem of obtaining pivot-free algorithm was studied in \cite{2},\cite{3} by S.Watt.
 He presented the algorithm that is based on the following identity for a
nonsingular matrix:\ $A^{-1}=(A^{T}A)^{-1}A^{T}$. Here $A^T$ is
the transposed matrix to $A$ and all principal minors of the
matrix $A^TA$ are nonzero. This method is useful for making an
efficient parallel program with the help of Strassen's fast
decomposition of inverse matrix for dense nonsingular matrix over
the field of zero characteristic  when  field elements are
represented by the float numbers. 
Other parallel matrix algorithms are developed in \cite{14} - \cite{17}.

This paper is devoted to the construction  of the pivot-free
matrix decomposition method in a common case of singular matrices
over a field of arbitrary characteristic. The decomposition will
be constructed in the form $LAU=E$, where $L$ and $U$ are lower
and upper triangular matrices, $E$ is a truncated permutation
matrix, which has the same rank as the matrix $A$. Then the
generalized Bruhat decomposition may be easy obtained using the
matrix $L$, $E$ and $U$. This algorithm has the same complexity as
matrix multiplication and does not require pivoting. For singular
matrices it allows to obtain a nonsingular block of the biggest
size, the echelon form and kernel of matrix. The preliminary
variants of this algorithm were developed in \cite{8} and \cite{9}.

\section{Preliminaries}
We introduce some notations that will be used in the following sections.

Let $F$ be a field, $F^{n\times n}$ be an  $n\times n$ matrix ring over $F$, $S_n$
be a permutation group of $n$ elements.
Let $P_n$ be a multiplicative semigroup in $F^{n\times n}$ consisting of
matrices $A$ having exactly $\rm{rank} (A)$ nonzero entries, all of them equal to $1$.
We call $P_n$ the permutation semigroup  because it contains the permutation group
of n elements  $S_n$ and all their truncated matrix.

The semigroup $D_n \subset P_n$ is formed by the diagonal matrices. So
$|D_n|$=$2^n$ and the identity matrix ${\mathbf I}$ is the identity element in $D_n$, $S_n$ and $P_n$.

Let $W_{i,j}\in P_n$ be a matrix, which has only one nonzero
element in the position $(i,j)$. For an arbitrary matrix $E$ of
$P_n$, which has the rank $n-s$ ($s=0,..n$) we shall denote by
$i_{\overline E}=\{i_1,..,i_s  \}$ the ordered set of zero row
numbers and $i_{\overline E}=\{i_1,..,i_s  \}$ the ordered set of
zero column numbers.

\begin{definition}
Let $E\in P_n$ be the matrix of the rank $n-s$, let $i_{\overline
E}=\{i_1,..,i_s  \}$ and $i_{\overline E}=\{i_1,..,i_s  \}$ are
the ordered set of zero row numbers and zero columns number of the
matrix $E$. Let us denote by $\overline E$ the matrix
$$
{\overline E} = \sum_{k=1,..s} W_{i_k,j_k}
$$
and call it the  {\it complimentary} matrix for $E$. For the case $s=0$ we put ${\overline E}=0$.
\end{definition}
It is easy to see that $\forall E\in P_n:\  {E+\overline E}\in
S_n$, and $\forall I\in D_n:\ {I+\overline I}=\mathbf I$.
Therefore the map  $I\mapsto \overline{I}={\mathbf I}-I$ is the
involution and we have $I\overline{I}=0$.  We can define the
partial order at $D_n$: $I<J \Leftrightarrow J-I\in D_n$.
For each matrix $E \in P_n$ we shall denote by
$$
I_E=EE^T  \hbox{ and }  J_E=E^TE
$$
the diagonal matrix: $I_E, J_E \in D_n$.
%Let us note that
The unit elements of the matrix $I_E$ show nonzero rows of the matrix $E$ and the
unit elements of the matrix
$J_E$ show nonzero columns of the matrix $E$.
Therefore we have several zero identities:
$$
E^T\overline{I}_E=\overline{I}_E E=E\overline{J}_E=\overline{J}_E E^T=0. \eqno(1)
$$

%If $I \in D_n$ then the product $IA$ has zero rows, which
%correspond to zero rows in $I$,  other rows are taken from $A$.
%Therefore we will use the notation $A=IA$, $(A=AJ)$, where $I,J\in D_n$
%and have the lowest rank,
%to specify zero rows (columns) of the matrix $A$.
For any pare $I,J\in D_n$  let us denote the subset of matrices $F^{n\times n}$
$$F_{I,J}^{n\times n}=\{B: B\in F^{n\times n},  IBJ=B  \}.  $$
We call them $(I,J)$-zero matrix. It is evident that  $F^{n\times
n}= F_{{\mathbf I},{\mathbf I}}^{n\times n}$, $0\in
\cup_{I,J}F_{I,J}^{n\times n}$ and  if $I_2<I_1$ and $J_2<J_1$
then
 $F_{I_2,J_2}^{n\times n}\subset F_{I_1,J_1}^{n\times n}$.

\begin{definition}
We shall call the factorization of the matrix $A\in
F_{I,J}^{n\times n}$
$$
A=L^{-1}EU^{-1},   \eqno(2)
$$
 $LEU$-decomposition  if $E \in P_n$, $L$ is a
nonsingular low triangular matrix, $U$ is an upper unitriangular
matrices and
%for the matrices $I_E=EE^T,\ J_E=E^TE$ holds the next property of $L$ and $U$:
$$
 L-\overline I_E \in F_{I,I_E}^{n\times n},\ \ U-\overline J_E \in F_{J_E,J}^{n\times n}.
\eqno(3)
$$
\end{definition}
If (2) is the $LEU$-decomposition we shall write
$$
(L,E,U)= \mathcal{LU}(A),
$$
\begin{sentence} Let $ (L,E,U)= \mathcal{LU}(A)$ be the $LEU$-decomposition of
matrix $A\in F_{I,J}^{n\times n}$
 then
$$
L=\overline I_E +ILI_E,\ U=\overline J_E + J_EUJ,\  E\in F_{I,J}^{n\times n},   \eqno(4)
$$
$$
L^{-1}=\overline I_E +L^{-1}I_E,\ U^{-1}=\overline J_E + J_E U^{-1}.
$$
\end{sentence}
\begin{proof}
The first and second equalities  follows from (3). To prove the
property of matrix $E$ we use the commutativity of diagonal
semigroup $D_n$: \linebreak
$$
E=LAU=(\overline I_E +ILI_E)IAJ(\overline J_E + J_EUJ)=I(\overline I_E +LI_EI)A(\overline J_E +
JJ_EU)J.
$$
To prove the property of matrix $L^{-1}$ let us consider the identity
$$
{\mathbf I}= L^{-1}L= L^{-1}(\overline I_E +LI_E)=L^{-1}\overline I_E + {\mathbf I} I_E
$$
Therefore
$L^{-1}\overline I_E=  \overline I_E$ and $L^{-1}=L^{-1}(\overline I_E+I_E)=\overline I_E+L^{-1} I_E$.
The prove of the matrix $U^{-1}$ property may be obtained similarly.
\end{proof}
Sentence 1 states the  property of matrix $E$, which may be
written in the form $I_E<I$, $J_E<J$. We shall call it
 {\it the property of immersion}.
on the other hand, each zero row of the matrix $E$ goes to the
unit column of matrix $L$ and each zero column of the matrix $E$
goes to the unit row of matrix $U$.

Let us denote by ${\mathcal E_n}$ the permutation matrix $W_{1,n}+W_{2,n-1}+..+W_{n,1}\in S_n$.
It is easy to see that if the matrix $A\in F^{n\times n}$ is low-(upper-) triangular, then 
the matrix ${\mathcal E_n} A {\mathcal E_n}$ is upper- (low-) triangular.
\begin{sentence}
 Let $(L,E,U)=\mathcal{LU}(A)$ be the $LEU$-decomposition of matrix $A\in F^{n\times
 n}$,
 then the
matrix ${\mathcal E_n}A$ has the generalized Bruhat decomposition
$V_1wV_2$ and
$$
V_1={\mathcal E_n}(L^{-1}-\overline I_E){\mathcal E_n},\  w={\mathcal E_n}(E+{\overline E}),\
 V_2=(U^{-1}-\overline J_E).
$$
\end{sentence}
\begin{proof}
As far as $L^{-1}$ is a low triangular matrix and $U^{-1}$ is an
upper triangular matrix we see that $V_1$ and $V_2$ are  upper
triangular matrices. Matrix $w$ is a product of  permutation
matrices so $w$ is a permutation matrix. One easily checks that
$V_1wV_2={\mathcal E_n}L^{-1}EU^{-1}={\mathcal E_n}A$.

\end{proof}

Examples.

For any matrix $I\in D_n$,  $E\in P_n$, $0 \neq a \in F$ the product $(aI+\overline I)
 I\ {\mathbf I}$ is a LEU decompositions of matrix $aI$ and the product
 $(aI_E+\overline I_E)  E\ {\mathbf I}$ is a  LEU decompositions of matrix  $aE$.

 %  %%%%%%%%%%%%%%%%%%%%%%%%%%%%%%%%%%%%%%%%%%%%%%%%%%%%%%%%%%%%%%%%%%%%%%%%%%%%%%%%%%%%%%%%%%%%%%%%%%%%%%%%%%%%%%%
{\section{Algorithm of $LEU$ decomposition}}

\begin{theorem}
For any matrix $A\in F^{n\times n}$ of size  $n=2^k,\ k\geq 0$ a
$LEU$-decomposition exists. For computing such decomposition it is
enough to compute 4 $LEU$-decompositions, 17
 multiplications and several permutations for the matrices of size  $n=2^{k-1}$.
\end{theorem}
\begin{proof}
For the matrix of size $1\times 1$, when $k=0$,  we can write the
following $LEU$ decompositions
$$
\mathcal{LU}(0)=(1,0,1) \hbox { and }\ {\mathcal{LU}}(a)=(a^{-1},1,1), \hbox{ if }\ a\neq 0.
$$
Let us assume that for any matrix of size $n$ we can write a $LEU$
decomposition and let us given matrix $A\in F^{2n\times 2n}_{I,J}$
has the size $2n$. We shall construct a $LEU$ decomposition of
matrix $A$.

First of all we shall divide the matrices $A$, $I$, $J$ and a
desired matrix $E$ into four equal blocks:
$$
A =
\left[\begin{array}{cc} A_{11} & A_{12} \\
A_{21} & A_{22} \end{array}\right],
I={\rm diag} (I_1,I_2), J={\rm diag} (J_1,J_2), E=\left[\begin{array}{cc} E_{11} & E_{12} \\
E_{21} & E_{22} \end{array}\right],
                                      \eqno(5)
$$
and denote
 $$I_{ij}=E_{ij} E_{ij}^T, \ \  J_{ij}=E_{ij}^T E_{ij}\ \  \forall i,j \in \{1,2\}.   \eqno(6)$$

Let
$$(L_{11}, E_{11}, U_{11})={\mathcal{LU}}(A_{11}), \eqno (7)$$
denote the matrices
$$
  Q =L_{11}A_{12},\ B = A_{21}U_{11}, \eqno(8)
$$
$$
A^{1}_{21}=  B \overline {J}_{11},\ A^{1}_{12}=\overline{I}_{11} Q,\ A^{1}_{22}=A_{22}-  B E^T_{11} Q.  \eqno(9)
$$
Let
$$(L_{12},E_{12},U_{12})={\mathcal{LU}}(A^1_{12}) \hbox{ and } (L_{21},E_{21},U_{21})={\mathcal{LU}}(A^1_{21}),  \eqno (10) $$
denote the matrices
$$
  G=L_{21}A^{1}_{22}U_{12},\
A^{2}_{22}=\overline{I}_{21} G\overline{J}_{12}.                             \eqno(11)
$$
Let us put
$$(L_{22},E_{22},U_{22})={\mathcal{LU}}(A^2_{22}),          \eqno (12) $$
and denote
$$
W=( G E^T_{12}L_{12} + L_{21}  B E_{11}^T), \
V=(U_{21}E_{21}^T G \overline{J}_{12} + E_{11}^T Q U_{12}),                              \eqno(13)
$$
$$ L = \left(\begin{array}{cc} L_{12}L_{11} & 0 \\
-L_{22}W L_{11} & L_{22}L_{21} \\
\end{array}\right),
 U = \left(\begin{array}{cc} U_{11}U_{21} &
 -U_{11}V U_{22} \\
0 & U_{12}U_{22} \\ \end{array}\right).                                \eqno(14)
$$
We have to prove that
$$
(L,E,U)={\mathcal{LU}}(A).   \eqno (15)
$$
As far as  $L_{11}, L_{12}, L_{21},L_{22}$ are  low triangular
nonsingular matrices and $U_{11}, U_{12}$, $U_{21}$, $U_{22}$ are upper
unitriangular matrices we can see in (10) that the matrix $L$ is a
low triangular nonsingular matrix and the matrix $U$ is  upper
unitriangular.

Let us show that $E\in P_{2n}$.
As far as $E_{11}, E_{12}, E_{21}, E_{22}\in P_n $ and $A_{11}={I}_{11}A_{11}{J}_{11}$,
$A^{1}_{21}=  B\overline {J}_{11}$, $A^{1}_{12}=\overline{I}_{11} Q$,
$A^{2}_{22}=\overline{I}_{21} G\overline{J}_{12}$ and due to the Sentence 1 we obtain $E_{11}={I}_{11}
 E_{11}{J}_{11}$,
$E_{21}=  E_{21} \overline {J}_{11}$, $E_{12}=\overline{I}_{11} E_{12}$,
$E_{22}=\overline{I}_{21} E_{22}\overline{J}_{12}$.

%Следовательно, все четыре блока матрицы $E$ имеют единичные элементы, расположенные в разных
%строках и столбцах матрицы $E$. Поэтому $E\in P_{2n}$. При этом
%выполняются следующие равенства
Therefore  the unit elements in  each of the four blocks of the
matrix $E$ are disposed in different rows and columns of the
matrix $E$. So $E\in P_{2n}$, and next identities hold

$$E_{11}E^T_{21}= E_{11}J_{21}= J_{11}E^T_{21}= J_{11}J_{21}= 0,   \eqno(16) $$
$$E^T_{12}E_{11}=E^T_{12}I_{11}=I_{12}E_{11}=I_{12}I_{11}=0, \eqno(17) $$
$$E_{12}E^T_{22}=E_{12}J_{22}=J_{12}E^T_{22}=J_{12}J_{22}=0, \eqno(18) $$
$$E^T_{22}E_{21}=E^T_{22}I_{21}=I_{22}E_{21}=I_{22}I_{21}=0. \eqno(19) $$

%Нужно показать, что имеет место свойство (3).

%Докажем сначала равенство $LAU=E$. Для этого нужно
%доказать справедливость следующих четырех равенств для блоков матрицы $E$:

We have to prove, that $E=LAU$. This equation in block form consists of four block equalities:
$$
\begin{array}{l}
 E_{11}=L_{12}L_{11}A_{11}U_{11}U_{21} ; \\
 E_{12}=L_{12}L_{11}(A_{12}U_{12} - A_{11}U_{11}V)U_{22} ; \\
E_{21}=L_{22}(L_{21}A_{21} - WL_{11}A_{11})U_{11}U_{21} ; \\
E_{22}=L_{22}((L_{21}A_{22}- WL_{11}A_{12})U_{12}-( L_{21}A_{21}- WL_{11}A_{11})U_{11}V)U_{22}.
\end{array}
\eqno(20)
$$
Therefore we have to prove these block equalities.

%Предворительно отметим, что  так как
Let us note, that from the identity
$A_{11}=I_{1}A_{11}J_{1}$ %то по свойству
and Sentence 1 we get
$$
L_{11}={\overline I}_{11}+I_{1}L_{11} I_{11},\ U_{11}={\overline J}_{11}+J_{11}U_{12} J_{1}. \eqno(21)
$$

The Sentence 1 together with equations
$A^{1}_{12}=\overline{I}_{11} L_{11}A_{12}$, $A^{1}_{21}=
A_{21}U_{11} \overline {J}_{11}$, $A^{2}_{22}=\overline{I}_{21}
L_{21}(A_{22}-  A_{21}U_{11} E^T_{11}
L_{11}A_{12})U_{12}\overline{J}_{12}$ give the next properties of
L- and U- blocks:
$$
\begin{array}{l}
L_{12}={\overline I}_{12}+{\overline I}_{11}I_{1} L_{12} I_{12}, \
U_{12}={\overline J}_{12}+J_{12}U_{12} J_{2},
\\
L_{21}={\overline I}_{21}+I_{2} L_{21} I_{21}, \ \ \ \
U_{21}={\overline J}_{21}+J_{21}U_{12} J_{1} {\overline J}_{11},
\\
L_{22}={\overline I}_{22}+{\overline I}_{21} I_2 L_{22} I_{22}, \  \
U_{22}={\overline J}_{22}+J_{22}U_{22}J_2 {\overline J}_{12}.
\end{array} \eqno(22)
$$
The following identities can be easy checked now
$$L_{12}E_{11}=E_{11},\ L_{12}I_{11}=I_{11}, \eqno (23)$$
$$E_{11}U_{21}=E_{11},\ J_{11}U_{21}=J_{11}, \eqno (24) $$
$$E_{12}U_{22}=E_{12},\ J_{12}U_{22}=J_{12}, \eqno (25)$$
$$L_{22}E_{21}=E_{21},\  L_{22}I_{21}=I_{21}.  \eqno (26)$$
We shall use the following equalities,
$$ L_{11}A_{11}U_{11}=E_{11}, L_{12}A_{12}^1U_{12}=E_{12}, L_{21}A_{21}^1U_{21}=E_{21}, L_{22}A_{22}^2U_{22}=E_{22}, \eqno (27)$$
which follows from (7),(10) and (12), the equality
$$
E_{11}V=I_{11}QU_{12},   \eqno (28)
$$
which follows from the definition of  the block $V$ in (13), (24),
(16) and (6), the equality
$$
WE_{11} =L_{21}BJ_{11}, \eqno (29)
$$
which follows from the definition of the block  $W$ in (13), (23),
(17) and (6).

1. The first equality of (20) follows from (27), (23) and (24).

2.The right   part of the second equality of (20) takes the form
$L_{12}({\mathbf I}-I_{11})Q U_{12}U_{22}$  due to (8), (27) and
(28). To prove the second equality we use the definition of the
blocks $B$ and $A_{12}^1$ in (8) and (9), then the second equality
in (27) and identity (25): $L_{12}({\mathbf I}-I_{11})Q
U_{12}U_{22}=L_{12}A_{12}^1 U_{12}U_{22}=E_{12}U_{22}=E_{12}.$

3.  The right  part of the third equality of (20) takes the form
$L_{22}L_{21}B({\mathbf I}-J_{11}) U_{21}$ due to definition of
the block $B$ (8), the first equality in (27) and (29). To prove
the third equality we use the definition of the blocks $A_{21}^1$
in  (9), then the third equality in (27) and identity (26):
$L_{22}L_{21}B{\overline J}_{11} U_{21}=L_{22}L_{21}A_{21}^1
U_{21}=L_{22}E_{21}=E_{21}$.

4. The identity
$$
E^T_{12}L_{12}=
E^T_{12}L_{12}(I_{11}+{\overline I}_{11})=E^T_{12}L_{12}{\overline I}_{11} \eqno (30)
$$
follows from (23) and (17).

% В силу тождества $E_{22}=L_{22}A^2_{22}U_{22}$, для доказательства четвертого равенства, очевидно, достаточно показать, что верно равенство
% $A^2_{22}=(L_{21}A_{22}- WL_{11}A_{12})U_{12}-( L_{21}A_{21}- WL_{11}A_{11})U_{11}V$.

We have to check that
$(L_{21}A_{22}- WL_{11}A_{12})U_{12}=(L_{21}A_{22}- ( G E^T_{12}L_{12} + L_{21}  B E_{11}^T)Q)U_{12}=
L_{21}(A_{22}- B E_{11}^TQ)U_{12}-G E^T_{12}L_{12}QU_{12}=
L_{21}A^1_{22}U_{12}-G E^T_{12}L_{12}{\overline I}_{11}QU_{12}=
G-G E^T_{12}L_{12}A_{12}^1U_{12}=G-GE^T_{12}E_{12}=G{\overline J}_{12}$,
using the definitions of the blocks $W$ in (13), $A_{22}^1$ and $A_{12}^1$ in (9), the identity (28),
the second equality in (27) and the definition (6).

We have to check that $-( L_{21}A_{21}- WL_{11}A_{11})U_{11}V=-(
L_{21}A_{21}U_{11}- WE_{11})V$ $=$ $(-L_{21}B +
L_{21}BJ_{11})V=-L_{21}B{\overline J}_{11}V= -L_{21}B{\overline
J}_{11}(U_{21}E_{21}^T G \overline{J}_{12} + E_{11}^T Q U_{12})=
-L_{21}A^1_{21}U_{21}E_{21}^T G \overline{J}_{12}=-I_{21} G
\overline{J}_{12}$, using the first equality in (27), the identity
(29), the definitions of the blocks $V$ in (13), (1), then the
third equality in (27) and definition (6).

To prove the forth equality we have to substitute obtained
expressions to the right  part of the fourth equality:
$$
L_{22}(G{\overline J}_{12}-I_{21} G \overline{J}_{12})U_{22}=L_{22}{\overline I}_{21} G \overline{J}_{12}U_{22}=L_{22}A^2_{22}U_{22}=E_{22}.
$$

For the completion of the proving of this theorem we have to
demonstrate the special form of the matrices $U$ and $L$:
$L-\overline I_E \in F_{I,I_E}$ and $U-\overline J_E \in
F_{J_E,J}$.

The matrix $L$ is invertible and $I_E<I$ therefore we have to prove
that $L=\overline I_E +ILI_E$, where
$I_E={\rm diag}(I_{11}+I_{12}, I_{21}+I_{22})$,
$\overline I_E={\rm diag} (\overline I_{11} \overline I_{12}, \overline I_{21} \overline I_{22})$,
$I={\rm diag}(I_{1}, I_{2})$.

This matrix equality for matrix $L$ (14) is equivalent to the four block equalities:
\par\noindent
 $$ L_{12}L_{11} = I_1L_{12}L_{11}(I_{11}+I_{12}) +  \overline I_{11}\overline I_{12}, \  0=I_1 0 (I_{21}+I_{22}), $$
$$-L_{22} W L_{11} =  -I_2L_{22} W L_{11}(I_{11}+I_{12}), \
L_{22}L_{21} = I_2 L_{22}L_{21}(I_{21}+I_{22}) + \overline I_{21}\overline I_{22}.
$$
To prove the first block equalities we have to multiply its left
part by the unit matrix in the form ${\mathbf I}=(I_1+{\overline
I}_1)$ from the left side and by the unit matrix in the form
${\mathbf I}=(I_{11}+I_{12})+\overline I_{11}\overline I_{12}$
from the left side. Then we use the following identities to obtain
in the left part the same expression as in the right part: $
L_{11}\overline I_{11}=\overline I_{11}$, $ L_{12}\overline
I_{12}=\overline I_{12}$, $ \overline I_{1}L_{12}L_{11}=\overline
I_{1}$, $ \overline I_{1}(I_{11}+I_{12})=0$.   The same idea may
be used for proving the last block equality, but we must use other
forms of unit matrix: ${\mathbf I}=(I_2+{\overline I}_2)$,
${\mathbf I}=(I_{21}+I_{22})+\overline I_{21}\overline I_{22}$.

The  second block equality is evident.

Let us prove the third block equality.
We have to multiply the left part of the third block equality by the unit matrix in the form
${\mathbf I}=(I_2+{\overline I}_2)$ from the left side and by the unit matrix in the form
${\mathbf I}=(I_{11}+I_{12})+\overline I_{11}\overline I_{12}$ from the right side.

The block $W$ is equal to the following expression by the
definition (13), (11) and (8):
$$
W=( L_{21}(A_{22}-  A_{21}U_{11} E^T_{11} Q)U_{12} E^T_{12}L_{12} + L_{21} A_{21}U_{11}  E_{11}^T).
$$
We have to use in the left part the equations
 ${\overline I}_2 L_{22} = {\overline I}_2$,
${\overline I}_2 L_{21} = {\overline I}_2$,
${\overline I}_2 A_{22} = 0$, ${\overline I}_2 A_{21} = 0$,
and
$L_{11} {\overline I}_{11}={\overline I}_{11}$,
$L_{12} {\overline I}_{12}={\overline I}_{12}$,
$E^T_{12}{\overline I}_{12}=0$, $E^T_{11}{\overline I}_{11}=0$.

The property of the matrix  $U$:  $U-\overline J_E \in F_{J_E,J}$
may be proved in the same way as the property of the matrix $L$.

\end{proof}

% %%%%%%%%%%%%%%%%%%%%%%%%%%%%%%%%%%%%%%%%%%%%%%%%%%%%%%%%%%%%%%%%%%%%%%%%%%%%%%%%%%%%%%%%%

\begin{theorem}
For any matrix $A$ of size $s, (s\geq 1)$,  an algorithm of
$LEU$-decompo\-sition which has the same complexity as matrix
multiplication exists.
\end{theorem}
\begin{proof}
We have proved an existence of $LEU$-decomposition for matrices of
size $2^k, k>0.$ Let $A\in F^{s\times s}_{I,J}$ be a matrix of
size $2^{k-1}<s< 2^k$, $A'$ be a matrix of size $2^k$, which has
in the left upper corner the submatrix equal  $A$ and all other
elements equal zero. We can construct $LEU$-decomposition of
matrix $A'$: $(L', E', U')={\mathcal{LU}}(A')$. According to the
Sentence 1 the product $L' A' U'=E'$ has the form
$$ \left(\begin{array}{cccc}
L &  0 \\
0 & {\mathbf I}
 \end{array}\right)
\left(\begin{array}{cccc}
A &  0 \\
0 & 0
 \end{array}\right)
 \left(\begin{array}{cccc}
U &  0 \\
0 &  {\mathbf I}
 \end{array}\right)=
 \left(\begin{array}{cccc}
E &  0 \\
0 &  0
 \end{array}\right)
$$
Therefore $LAU=E$ is a $LEU$ decomposition of matrix $A$.

The total amount of matrix multiplications in (7)-(15) is equal to
17 and total amount of recursive calls is equal to 4. We do not
consider multiplications of the permutation matrices, we can do these multiplications due to permutation of pointers for the blocks which are disposed at the
lockal processors.

We can compute the decomposition  of the second order matrix  by
means of 5 multiplicative operations. Therefore  we get the
following recurrent equality for complexity
$$
t(n)=4 t(n/2)+ 17 M(n/2), t(2)=5.
$$
Let $\gamma$ and $\beta$ be constants, $3\geq\beta>2$, and let
  $M(n)= \gamma n^{\beta} + o(n^{\beta})$ be the number of multiplication
  operations in one $n\times n$ matrix multiplication.

After summation from  $n=2^k$ to $2^1$ we obtain
$$
17  \gamma(4^0 2^{\beta(k-1)} + \ldots + 4^{k-2}2^{\beta 1})+4^{k-2}
5 = 17 \gamma\frac{n^{\beta}-2^{\beta-2}n^2}{2^{\beta}-4} +
\frac{5}{16}n^2.
$$
Therefore the complexity of the decomposition is
$$
\sim
\frac{ 17 \gamma n^{\beta}}{2^{\beta}-4}
$$
\end{proof}

If $A$ is an invertible matrix, then $A^{-1}=UE^TL$ and a
recursive block algorithm of matrix inversion is written in the
expressions (7)-(15). This algorithm has the complexity of matrix
multiplications.

%\subsection{Solving of indefinite systems of linear equations}
%

\section{Conclusion}

An algorithms for finding the generalized Bruhat decomposition and
matrix inversion are described. These algorithms have the same
complexity as matrix multiplication and do not require pivoting.
For singular matrices they allow to obtain a nonsingular block of
the biggest size. These algorithms may be used in any field,
including  real and complex numbers, finite fields and their
extensions.

The proposed algorithms are pivot-free, and do not change the matrix block
structure. So they are suitable for parallel hardware implementation.

% ---- Bibliography ----

\end{document}